\documentclass[%
 reprint, amsmath,amssymb,superscriptaddress,
 aps]{revtex4-1}
\usepackage{enumerate}
\usepackage{amsmath}
\usepackage{booktabs}
\usepackage{array}
\usepackage{amsthm}
\usepackage{float}
\usepackage{dcolumn}
\usepackage[colorlinks=true, citecolor=blue, urlcolor=blue ]{hyperref}
\usepackage{bm}
\usepackage{longtable}
\usepackage{graphicx}
\usepackage{epsfig}
\usepackage{epstopdf}
\usepackage{amssymb}
\usepackage{color}
\usepackage{makecell}

\allowdisplaybreaks[4]

\begin{document}

\preprint{APS/123-QED}

\title{Genuinely nonlocal sets without entanglement in multipartite systems}

\author{Ying-Ying Lu}
\affiliation{School of Mathematical Sciences, Hebei Normal University, Shijiazhuang 050024, China}
\author{Hai-Qing Cao}
\affiliation{School of Mathematical Sciences, Hebei Normal University, Shijiazhuang 050024, China}
\author{Hui-Juan Zuo}
\email{huijuanzuo@163.com}
\affiliation{School of Mathematical Sciences, Hebei Normal University, Shijiazhuang 050024, China}
\affiliation{Hebei Key Laboratory of Computational Mathematics and Applications, Shijiazhuang, 050024, China}
\affiliation{Hebei International Joint Research Center for Mathematics and Interdisciplinary Science, Shijiazhuang 050024, China}
\author{Shao-Ming Fei}
\affiliation{School of Mathematical Sciences, Capital Normal University, Beijing 100048, China}


\begin{abstract}
A set of multipartite orthogonal states is genuinely nonlocal if it is locally indistinguishable in every bipartition of the subsystems. If the set is locally reducible, we say it has genuine nonlocality of type \uppercase\expandafter{\romannumeral 1}. Otherwise, we say it has genuine nonlocality of type \uppercase\expandafter{\romannumeral 2}. Due to the complexity of the problem, the construction of genuinely nonlocal sets in general multipartite systems has not been completely solved so far. In this paper, we first provide a nonlocal set of product states in bipartite systems. We obtain a genuinely nonlocal set of type~\uppercase\expandafter{\romannumeral 1} without entanglement in general $n$-partite systems $\otimes^{n}_{i=1}\mathbb{C}^{d_{i}}$ $[3\leq (d_{1}-1)\leq d_{2}\leq \cdots\leq d_{n},n\geq3]$. Then we present two constructions with genuine nonlocality of type~\uppercase\expandafter{\romannumeral 2} in $\mathbb{C}^{d_{1}}\otimes\mathbb{C}^{d_{2}}\otimes\mathbb{C}^{d_{3}}$ $(3\leq d_{1}\leq d_{2}\leq d_{3})$ and $\otimes^{n}_{i=1}\mathbb{C}^{d_{i}}$ $(3\leq d_{1}\leq d_{2}\leq \cdots\leq d_{n},n\geq4)$. Our results further positively answer the open problem that there does exist a genuinely nonlocal set of type~\uppercase\expandafter{\romannumeral2} in multipartite systems [M. S. Li, Y. L. Wang, F. Shi, and M. H. Yung, J. Phys. A: Math. Theor. 54, 445301 (2021)] and highlight its related applications in quantum information processing.

\end{abstract}

\pacs{Valid PACS appear here}
\maketitle


\section{\label{sec:level1}Introduction\protect}

 Quantum nonlocality is one of the most fascinating phenomena in quantum mechanics. The nonlocality from the perspective of state discrimination exhibits the property of being an orthogonal state set. This kind of nonlocality is not the same as the Bell nonlocality \cite{Bell nonlocality}. A set of orthogonal quantum states is {\it locally indistinguishable or nonlocal} if it is not possible to perfectly distinguish the states by any sequence of local operations and classical communications (LOCC). In 1999, Bennett \textit{et al}. \cite{r1} first constructed a locally indistinguishable orthogonal product basis in $\mathbb{C}^{3}\otimes\mathbb{C}^{3}$, which displays the phenomenon of quantum nonlocality without entanglement. Once the information is encoded in a nonlocal set, it can be accessible only by a global measurement. This feature has significant applications in quantum data hiding and quantum secret sharing \cite{app1,app2,app3,app4,app5,app6,app7,app8}. In recent years, great progress has been made in nonlocality theory \cite{r10,r12,r18,r19,r20,r22,r23,r24,r25,r27,r31,r32,Cao stable,Cao distinguish}.

A set of orthogonal quantum states is {\it locally irreducible} if it is not possible to eliminate one or more quantum states from the set by nontrivial orthogonality-preserving local measurements (OPLMs). Halder \textit{et al}. \cite{r4} introduced the concept of local irreducibility and gave a stronger manifestation of nonlocality in multipartite systems, namely, the strong nonlocality, which has been widely investigated \cite{r15,r16,r17,r21,r26,add1,add2,add3,add4,add5,add6}. A set of multipartite orthogonal quantum states is {\it genuinely nonlocal} if it is locally indistinguishable in every bipartition. Rout \textit{et al}. \cite{r13} proposed the concept and the classification of genuinely nonlocal product bases and constructed different types of genuinely nonlocal product bases in $\mathbb{C}^{3}\otimes\mathbb{C}^{3}\otimes\mathbb{C}^{3}$ and $\mathbb{C}^{4}\otimes\mathbb{C}^{4}\otimes\mathbb{C}^{4}$. In particular, when all the parties are separated, if a genuinely nonlocal orthogonal state set is locally reducible, we say it has genuine nonlocality of type \uppercase\expandafter{\romannumeral 1}; otherwise, we say it has genuine nonlocality of type \uppercase\expandafter{\romannumeral 2}. The strength of genuine nonlocality mentioned here is stronger than the indistinguishability-based nonlocality but weaker than the irreducibility-based strong nonlocality. In 2021, Li \textit{et al}. \cite{r9} presented the genuinely nonlocal orthogonal product sets of size $2x+4y+2z-8$ in $\mathbb{C}^{x}\otimes\mathbb{C}^{y}\otimes\mathbb{C}^{z}$ $(x,z\geq3,y\geq4)$. They also gave the general construction of genuinely nonlocal product states in \textit{N}-partite systems for the first time. In Ref. \cite{r14}, Rout \textit{et al}. constructed a set of genuinely nonlocal orthogonal product states in $\mathbb{C}^{4}\otimes\mathbb{C}^{3}\otimes\mathbb{C}^{3}$ and extended  it to  $m+1$-partite systems of $\mathbb{C}^{m+2}\otimes (\mathbb{C}^{3})^{\otimes m}$. In addition, they also gave the construction of a genuinely nonlocal orthogonal product set in $\mathbb{C}^{6}\otimes\mathbb{C}^{5}\otimes\mathbb{C}^{5}$. Later, in Refs. \cite{Xiong2023,Xiong202402,Xiong202403}, the authors constructed sets of multipartite entangled states which were genuinely nonlocal.

In this paper, we first provide the locally indistinguishable set of orthogonal product states with size $2d_{2}-1$ in $\mathbb{C}^{d_{1}}\otimes\mathbb{C}^{d_{2}}$ $(3\leq d_{1}\leq d_{2})$. It is worth mentioning that our structure greatly optimizes the known results. Based on the nonlocal set, we present an orthogonal product state set with genuine nonlocality of type \uppercase\expandafter{\romannumeral 1} in $\mathbb{C}^{d_{1}}\otimes\mathbb{C}^{d_{2}}\otimes\mathbb{C}^{d_{3}}$ $[3\leq (d_{1}-1)\leq d_{2}\leq d_{3}]$ and generalize the construction to \textit{n}-partite systems of $\otimes^{n}_{i=1}\mathbb{C}^{d_{i}}$ $[3\leq (d_{1}-1)\leq d_{2}\leq \cdots\leq d_{n},n\geq3]$. Then we construct type-\uppercase\expandafter{\romannumeral 2} genuinely nonlocal sets without entanglement in $\mathbb{C}^{d_{1}}\otimes\mathbb{C}^{d_{2}}\otimes\mathbb{C}^{d{_3}}$ $(3\leq d_{1}\leq d_{2}\leq d_{3})$ and $\otimes^{n}_{i=1}\mathbb{C}^{d_{i}}$ $(3\leq d_{1}\leq d_{2}\leq\cdots\leq d_{n},n\geq4)$, respectively. Finally, we draw our conclusions and put forward some questions for further consideration.

\theoremstyle{remark}
\newtheorem{definition}{\indent Definition}
\newtheorem{lemma}{\indent Lemma}
\newtheorem{theorem}{\indent Theorem}
\newtheorem{proposition}{\indent Proposition}
\newtheorem{corollary}{\indent Corollary}
\newtheorem{example}{\indent Example}
\def\QEDclosed{\mbox{\rule[0pt]{1.3ex}{1.3ex}}}|
\def\QED{\QEDclosed}
\def\proof{\noindent{\indent\em Proof}.}
\def\endproof{\hspace*{\fill}~\QED\par\endtrivlist\unskip}

\section{Preliminaries}

Throughout this paper, we only consider pure states and we do not normalize states for simplicity. For any positive integer $d\geq2$, $\mathbb{Z}_{d}=$ \{$0,1,\cdots,d-1$\}. Suppose that $\mathcal{H}$ is a \textit{d}-dimensional Hilbert space and \{$|0\rangle,|1\rangle,\cdots,|d-1\rangle\}$ is the computational basis of $\mathcal{H}$. For convenience, we denote $|i_{1} \pm i_{2} \pm \cdots \pm i_{n}\rangle$ the state $\frac{1}{\sqrt{n}}(|i_{1}\rangle \pm |i_{2}\rangle \pm \cdots \pm |i_{n}\rangle)$. For each integer $d\geq2$, $\omega_{d}=e^{\frac{2\pi\sqrt{-1}}{d}}$, i.e., a primitive $d$th root of unit. A positive operator-valued measurement (POVM) on Hilbert space  $\mathcal{H}$ is a set of semidefinite operators \{$E_{m}=M_{m}^{\dagger}M_{m}$\} such that $\sum_{m}E_{m}=\mathbb{I}_\mathcal{H}$, where each $E_{m}$ is called a POVM element, and $\mathbb{I}_\mathcal{H}$ is the identity operator on $\mathcal{H}$.

\begin{definition} Locally indistinguishable \cite{r1}. A set of orthogonal states is locally indistinguishable or nonlocal if it cannot be perfectly distinguished under local operations and classical communications (LOCC).
\end{definition}

\begin{definition} Orthogonality-preserving local measurement. A measurement of a mutually orthogonal quantum state set is called an orthogonality-preserving local measurement (OPLM), if the postmeasurement states remain mutually orthogonal. Furthermore, such a measurement is called trivial if its POVM elements are all proportional to the identity operator; otherwise, it is nontrivial.
\end{definition}

\begin{definition} Locally irreducible \cite{r4}. A set of orthogonal quantum states is locally irreducible if it is not possible to eliminate one or more quantum states from the set by nontrivial OPLMs.
\end{definition}

It is not difficult to show that, while a locally irreducible set is locally indistinguishable, the opposite is not true in general. To prove local indistinguishability (or local irreducibility), it is sufficient to verify that ``orthogonality-preserving local measurements are only trivial measurements" \cite{r19}.

\begin{definition} Genuinely nonlocal \cite{r13}. A set of multipartite orthogonal states is genuinely nonlocal if it is locally indistinguishable in every bipartition of the subsystems.
\end{definition}

For a genuinely nonlocal set, if the set is locally reducible, we say it has genuine nonlocality of type \uppercase\expandafter{\romannumeral 1}. Otherwise, we say it has genuine nonlocality of type \uppercase\expandafter{\romannumeral 2}.

\section{Orthogonal product sets with genuine nonlocality of type \expandafter{\romannumeral 1}}

In this section, we construct type-\uppercase\expandafter{\romannumeral 1} genuinely nonlocal sets without entanglement, which are locally reducible.
In the following we denote \textit{A}, \textit{B}, and \textit{C} as the three subsystems of a tripartite system associated with Alice, Bob, and Charlie, respectively, and we denote $A_{1},A_{2},\cdots,A_{n}$ as the \textit{n} subsystems of an \textit{n}-partite system associated with, without confusion, $A_{1},A_{2},\cdots,A_{n}$.

\subsection{Nonlocal sets in $\mathbb{C}^{d_{1}}\otimes \mathbb{C}^{d_{2}}$}

In Ref. \cite{r25}, Yu and Oh constructed the nonlocal set of orthogonal product states in $\mathbb{C}^{d}\otimes\mathbb{C}^{d}$. Inspired by the construction in Ref. \cite{r25}, we first construct a set of orthogonal product states with size of $2d_{2}-1$ in $\mathbb{C}^{d_{1}}\otimes\mathbb{C}^{d_{2}}$ $(3\leq d_{1}\leq d_{2})$ that is indistinguishable under LOCC.

\begin{proposition}\label{eg:35}
In $\mathbb{C}^{3}\otimes\mathbb{C}^{5}$, the set $\mathcal{S}_{1}$ consisting of the following nine orthogonal product states cannot be perfectly distinguished under LOCC:
\begin{equation}
\begin{aligned}
 |\phi_{1}\rangle=&|1\rangle_{A}|0-1\rangle_{B},\\
 |\phi_{2}\rangle=&|2\rangle_{A}|0-2\rangle_{B},\\
 |\phi_{3}\rangle=&|0-1\rangle_{A}|2\rangle_{B},\\
 |\phi_{4}\rangle=&|0-2\rangle_{A}|1\rangle_{B},\\
 |\phi_{5}\rangle=&|0-1\rangle_{A}|3\rangle_{B},\\
 |\phi_{6}\rangle=&|0-1\rangle_{A}|4\rangle_{B},\\
 |\phi_{7}\rangle=&|0+1\rangle_{A}(|2\rangle_{B}+\omega_{3}|3\rangle_{B}+\omega_{3}^{2}|4\rangle_{B}),\\
 |\phi_{8}\rangle=&|0+1\rangle_{A}(|2\rangle_{B}+\omega_{3}^{2}|3\rangle_{B}+\omega_{3}|4\rangle_{B}),\\
 |\phi_{9}\rangle=&|0+1+2\rangle_{A}|0+1+2+3+4\rangle_{B}.
\end{aligned}
\end{equation}
\end{proposition}

The proof is given in Appendix \ref{app:1}. Based on the above proposition, we give a general construction of a locally indistinguishable set in $\mathbb{C}^{d_{1}}\otimes\mathbb{C}^{d_{2}}$.

\begin{lemma}\label{th:d1d2}
In $\mathbb{C}^{d_{1}}\otimes\mathbb{C}^{d_{2}}$ $(3\leq d_{1}\leq d_{2})$, the set $\mathcal{S}_{2}$ consisting of the following $2d_{2}-1$ orthogonal product states cannot be perfectly distinguished under LOCC:
\begin{equation}\label{eq:d1d2}
\begin{aligned}
|\phi_{i}\rangle=&|i\rangle_{A}|0-i\rangle_{B},~ 1\leq i\leq d_{1}-1,\\
|\phi_{i+(d_{1}-1)}\rangle=&|0-i\rangle_{A}|m\rangle_{B},~ 1\leq i\leq d_{1}-2,~ m=i+1;\\&i=d_{1}-1,~m=1,\\
|\phi_{j+(d_{1}-1)}\rangle=&|0-1\rangle_{A}|j\rangle_{B},~ d_{1}\leq j\leq d_{2}-1,\\
|\phi_{s+(d_{1}+d_{2}-2)}\rangle=&|0+1\rangle_{A} (|2\rangle_{B}\\
&+\sum\limits_{t=1}^{d_{2}-d_{1}}\omega_{d_{2}-d_{1}+1}^{st} |t+d_{1}-1\rangle_{B}), \\&1\leq s\leq d_{2}-d_{1},\\
|\phi_{2d_{2}-1}\rangle=&|0+1+\cdots+(d_{1}-1)\rangle_{A} \\&|0+1+\cdots+(d_{2}-1)\rangle_{B}.\\
\end{aligned}
\end{equation}
\end{lemma}

The proof of the lemma is given in Appendix \ref{app:2}.

\subsection{Genuinely nonlocal sets in $\mathbb{C}^{d_{1}}\otimes\mathbb{C}^{d_{2}}\otimes\mathbb{C}^{d_{3}}$ $[3\leq (d_{1}-1)\leq d_{2}\leq d_{3}]$}

We now construct a genuinely nonlocal set of type \uppercase\expandafter{\romannumeral 1} in $\mathbb{C}^{d_{1}}\otimes\mathbb{C}^{d_{2}}\otimes\mathbb{C}^{d_{3}}$ $[3\leq (d_{1}-1)\leq d_{2}\leq d_{3}]$ based on Lemma $\ref{th:d1d2}$. We first give a proposition.

\begin{proposition}
The following $18$ orthogonal product states in $\mathbb{C}^4\otimes \mathbb{C}^4\otimes \mathbb{C}^6$ are genuinely nonlocal:
\begin{equation}
\begin{aligned}
 |\phi_{1}\rangle=&|1\rangle_{A}|0-1\rangle_{B}|0\rangle_{C},\\
 |\phi_{2}\rangle=&|2\rangle_{A}|0-2\rangle_{B}|0\rangle_{C},\\
 |\phi_{3}\rangle=&|0-1\rangle_{A}|2\rangle_{B}|0\rangle_{C},\\
 |\phi_{4}\rangle=&|0-2\rangle_{A}|1\rangle_{B}|0\rangle_{C},\\
 |\phi_{5}\rangle=&|0-1\rangle_{A}|3\rangle_{B}|0\rangle_{C},\\
 |\phi_{6}\rangle=&|0+1\rangle_{A}|2-3\rangle_{B}|0\rangle_{C},\\
 |\phi_{7}\rangle=&|0+1+2\rangle_{A}|0+1+2+3\rangle_{B}|0\rangle_{C},\\
 |\phi_{8}\rangle=&|3\rangle_{A}|1\rangle_{B}|0-1\rangle_{C},\\
 |\phi_{9}\rangle=&|3\rangle_{A}|2\rangle_{B}|0-2\rangle_{C},\\
 |\phi_{10}\rangle=&|3\rangle_{A}|3\rangle_{B}|0-3\rangle_{C},\\
 |\phi_{11}\rangle=&|3\rangle_{A}|0-1\rangle_{B}|2\rangle_{C},\\
 |\phi_{12}\rangle=&|3\rangle_{A}|0-2\rangle_{B}|3\rangle_{C},\\
 |\phi_{13}\rangle=&|3\rangle_{A}|0-3\rangle_{B}|1\rangle_{C},\\
 |\phi_{14}\rangle=&|3\rangle_{A}|0-1\rangle_{B}|4\rangle_{C},\\
 |\phi_{15}\rangle=&|3\rangle_{A}|0-1\rangle_{B}|5\rangle_{C},\\
 |\phi_{16}\rangle=&|3\rangle_{A}|0+1\rangle_{B}(|2\rangle_{C}+\omega_{3}|4\rangle_{C} +\omega_{3}^{2}|5\rangle_{C}),\\
 |\phi_{17}\rangle=&|3\rangle_{A}|0+1\rangle_{B}(|2\rangle_{C}+\omega_{3}^{2}|4\rangle_{C} +\omega_{3}|5\rangle_{C}),\\
 |\phi_{18}\rangle=&|3\rangle_{A}|0+1+2+3\rangle_{B}|0+1+\cdots+5\rangle_{C}.\\
 \end{aligned}
 \end{equation}
\end{proposition}

\begin{proof} By Lemma \ref{th:d1d2}, $\mathcal{A}_{1}=\{ |1\rangle_{A}|0-1\rangle_{B}, |2\rangle_{A}|0-2\rangle_{B}, |0-1\rangle_{A}|2\rangle_{B}, |0-2\rangle_{A}|1\rangle_{B}, |0-1\rangle_{A}|3\rangle_{B}, |0+1\rangle_{A}|2-3\rangle_{B}, |0+1+2\rangle_{A}|0+1+2+3\rangle_{B}\}$ is a locally indistinguishable orthogonal product set of system $AB$ with size $7$ in $\mathbb{C}^3\otimes \mathbb{C}^4$. $\mathcal{B}_{1}=\{ |1\rangle_{B}|0-1\rangle_{C}, |2\rangle_{B}|0-2\rangle_{C}, |3\rangle_{B}|0-3\rangle_{C}, |0-1\rangle_{B}|2\rangle_{C}, |0-2\rangle_{B}|3\rangle_{C}, |0-3\rangle_{B}|1\rangle_{C}, |0-1\rangle_{B}|4\rangle_{C}, |0-1\rangle_{B}|5\rangle_{C},
|0+1\rangle_{B}(|2\rangle_{C}+\omega_{3}|4\rangle_{C} +\omega_{3}^{2}|5\rangle_{C}), |0+1\rangle_{B}(|2\rangle_{C}+\omega_{3}^{2}|4\rangle_{C} +\omega_{3}|5\rangle_{C} ),
|0+1+2+3\rangle_{B}|0+1+\cdots+5\rangle_{C}\}$ is a locally indistinguishable orthogonal product set of system $BC$ with size $11$ in $\mathbb{C}^4\otimes \mathbb{C}^6$.

Note that the tensors of the vectors in $\mathcal{A}_{1}$ and $|0\rangle_{C}$ are just the states $\{ |\phi_{i}\rangle\} _{i=1}^{7}$. Hence, both Alice and Bob can only perform trivial OPLMs, and then $AC$ and $BC$ cannot locally distinguish these states, which leads to local indistinguishability in partitions $B|AC$ and $A|BC$.

Similarly, the tensors of the vectors in $\mathcal{B}_{1}$ and $|3\rangle_{A}$ are just the states $\{ |\phi_{i}\rangle\} _{i=8}^{18}$. As the vectors in the subsystem $A$ are the same, these $11$ states $\{ |\phi_{i}\rangle\} _{i=8}^{18}$ cannot be distinguished in partitions $C|AB$ and $B|CA$ under LOCC. Thus, all the $18$ orthogonal product states cannot be perfectly distinguished in any bipartition. This completes the proof.
\end{proof}

From the above proposition, we obtain the following general construction.

First, we give an orthogonal product set $\mathcal{A}$ of size $2d_{2}-1$ by Eq. (\ref{eq:d1d2}):
\begin{equation}\label{eq:14}
\begin{aligned}
|\phi_{i}\rangle=&|i\rangle_{A}|0-i\rangle_{B}|0\rangle_{C},~1\leq i\leq d_{1}-2,\\
|\phi_{i+(d_{1}-2)}\rangle=&|0-i\rangle_{A}|m\rangle_{B}|0\rangle_{C}, ~1\leq i\leq d_{1}-3, \\&m=i+1;~i=d_{1}-2,~m=1,\\
|\phi_{j+(d_{1}-2)}\rangle=&|0-1\rangle_{A}|j\rangle_{B}|0\rangle_{C},~d_{1}-1\leq j\leq d_{2}-1,\\
|\phi_{s+(d_{1}+d_{2}-3)}\rangle=&|0+1\rangle_{A}|\eta\rangle_{B}|0\rangle_{C},~1\leq s\leq d_{2}-d_{1}+1,\\
|\phi_{2d_{2}-1}\rangle=&|0+1+\cdots+(d_{1}-2)\rangle_{A} \\&|0+1+\cdots+(d_{2}-1)\rangle_{B}|0\rangle_{C}.
\end{aligned}
\end{equation}
where $|\eta\rangle_{B}=|2\rangle_{B}+\sum\limits_{t=1}^{d_{2}-d_{1}+1}\omega_{d_{2}-d_{1}+2}^{st}|t+d_{1}-2\rangle_{B}$.

Then we give an orthogonal product set $\mathcal{B}$ of size $2d_{3}-1$ by Eq. (\ref{eq:d1d2}):
\begin{equation}\label{eq:15}
\begin{aligned}
|\phi_{i+2d_{2}-1}\rangle=&|d_{1}-1\rangle_{A}|i\rangle_{B}|0-i\rangle_{C},~1\leq i\leq d_{2}-1,\\
|\phi_{i+3d_{2}-2}\rangle=&|d_{1}-1\rangle_{A}|0-i\rangle_{B}|m\rangle_{C},~1\leq i\leq d_{2}-2,\\&m=i+1;~i=d_{2}-1,~m=1,\\
|\phi_{j+3d_{2}-2}\rangle=&|d_{1}-1\rangle_{A}|0-1\rangle_{B}|j\rangle_{C},~d_{2}\leq j\leq d_{3}-1,\\
|\phi_{s+3d_{2}+d_{3}-3}\rangle=&|d_{1}-1\rangle_{A}|0+1\rangle_{B}|\eta\rangle_{C},~1\leq s\leq d_{3}-d_{2},\\
|\phi_{2(d_{2}+d_{3})-2}\rangle=&|d_{1}-1\rangle_{A}|0+1+\cdots+(d_{2}-1)\rangle_{B}\\&|0+1+\cdots+(d_{3}-1)\rangle_{C},\\
\end{aligned}
\end{equation}
where $|\eta\rangle_{C}=|2\rangle_{C}+\sum\limits_{t=1}^{d_{3}-d_{2}}\omega_{d_{3}-d_{2}+1}^{st}|t+d_{2}-1\rangle_{C}$.

\begin{theorem}\label{The:1}
In $\mathbb{C}^{d_{1}}\otimes\mathbb{C}^{d_{2}}\otimes\mathbb{C}^{d{_3}}$ $[3\leq (d_{1}-1)\leq d_{2}\leq d_{3}]$, the $2(d_{2}+d_{3})-2$ orthogonal product states given by Eqs. (\ref{eq:14}) and (\ref{eq:15}) are genuinely nonlocal.
\end{theorem}

\begin{proof} First, the $2(d_{2}+d_{3})-2$ orthogonal product states in Eqs.(\ref{eq:14}) and (\ref{eq:15}) cannot be perfectly distinguished by LOCC when all the parties are separated.

The $2d_{2}-1$ states in $\mathcal{A}$ have the same part in $C$. They are locally indistinguishable in the joint systems $AC$ and $BC$. In fact, both Alice and Bob can only perform trivial OPLMs. Consequently, $\mathcal{A}$ is locally indistinguishable in bipartitions $B|AC$ and $A|BC$.

The $2d_{3}-1$ states in $\mathcal{B}$ are locally indistinguishable in bipartitions $C|AB$ and $B|CA$ as they are the same in part $A$, and then the $2(d_{2}+d_{3})-2$ orthogonal product states cannot be perfectly distinguished under LOCC in every bipartition, which means that the set is genuinely nonlocal. Moreover, when Alice performs the measurements $\{M^{A}_{1}=|d_{1}-1\rangle_{A} \langle d_{1}-1|, M^{A}_{2}=I_{A}-M^{A}_{1}\}$, the whole set can be locally reduced to two disjoint subsets \{$|\phi_{1}\rangle, \cdots, |\phi_{2d_{2}-1}\rangle$\} and \{$|\phi_{2d_{2}}\rangle, \cdots, |\phi_{2(d_{2}+d_{3})-2}\rangle$\}. Thus they have genuine nonlocality of type \uppercase\expandafter{\romannumeral 1}.
\end{proof}

\subsection{Genuinely nonlocal sets in $\otimes^{n}_{i=1}\mathbb{C}^{d_{i}}$ $[3\leq (d_{1}-1)\leq d_{2}\leq \cdots\leq d_{n},n\geq3]$}

Next we consider the construction of orthogonal product states with genuine nonlocality of type \uppercase\expandafter{\romannumeral 1} in $\otimes^{n}_{i=1}\mathbb{C}^{d_{i}}$ $[3\leq (d_{1}-1)\leq d_{2}\leq \cdots\leq d_{n},n\geq3]$.

\begin{theorem}\label{The:2}
In $\otimes^{n}_{i=1}\mathbb{C}^{d_{i}}$ $[3\leq (d_{1}-1)\leq d_{2}\leq\cdots\leq d_{n}, n\geq3]$, the following $\sum_{i=2}^{n}(2d_{i}-1)$ orthogonal product states have genuine nonlocality of type \uppercase\expandafter{\romannumeral 1}:
\begin{equation}
\begin{aligned}
G_{1}=&\{|x_{m_{1}}\rangle_{A_{1}}|y_{m_{1}}\rangle_{A_{2}}|0\rangle_{A_{3}}|0
\rangle_{A_{4}}\cdots\\&|0\rangle_{A_{n-3}}|0\rangle_{A_{n-2}}
|0\rangle_{A_{n-1}}|1\rangle_{A_{n}}\},\\
G_{2}=&\{|1\rangle_{A_{1}}|x_{j_{2}}\rangle_{A_{2}}|z_{j_{2}}\rangle_{A_{3}}
|0\rangle_{A_{4}}\cdots\\&|0\rangle_{A_{n-3}}|0\rangle_{A_{n-2}}
|0\rangle_{A_{n-1}}|0\rangle_{A_{n}}\},\\
G_{3}=&\{|0\rangle_{A_{1}}|x_{j_{3}}\rangle_{A_{2}}
|1\rangle_{A_{3}}|z_{j_{3}}\rangle_{A_{4}}\cdots\\
&|0\rangle_{A_{n-3}}|0\rangle_{A_{n-2}}|0\rangle_{A_{n-1}}|0\rangle_{A_{n}}\},\\
& ~~~~~~~~~~~~~~~~~~\vdots~~~~~~~~~~~~\\
G_{n-2}=&\{|0\rangle_{A_{1}}|x_{j_{n-2}}\rangle_{A_{2}}
|0\rangle_{A_{3}}|0\rangle_{A_{4}}\cdots\\&|0\rangle_{A_{n-3}}
|1\rangle_{A_{n-2}}|z_{j_{n-2}}\rangle_{A_{n-1}}|0\rangle_{A_{n}}\},\\
G_{n-1}=&\{|d_{1}-1\rangle_{A_{1}}|x_{j_{n-1}}\rangle_{A_{2}}
|0\rangle_{A_{3}}|0\rangle_{A_{4}}\cdots\\&|0\rangle_{A_{n-3}}
|0\rangle_{A_{n-2}}|1\rangle_{A_{n-1}}|z_{j_{n-1}}\rangle_{A_{n}}\},\\
\end{aligned}
\end{equation}
where $\{|x_{m_{1}}\rangle|y_{m_{1}}\rangle,m_{1}=1,2,\cdots,2d_{2}-1\}$ are $2d_{2}-1$ locally indistinguishable orthogonal product states in $\mathbb{C}^{d_{1}-1}\otimes \mathbb{C}^{d_{2}}$. And $\{|x_{j_{k-1}}\rangle|z_{j_{k-1}}\rangle,j_{k-1}=1,2,\cdots,2d_{k}-1\}$ are $2d_{k}-1$ locally indistinguishable orthogonal product states in $\mathbb{C}^{d_{2}}\otimes \mathbb{C}^{d_{k}}$ $(k=3,\cdots, n)$.
\end{theorem}

\begin{proof} In fact, as far as $G_{1}$ is concerned, both $A_{1}$ and $A_{2}$ can only perform trivial OPLMs by Lemma \ref{th:d1d2}. Thus perfect discrimination of $G_{1}$ can be achieved only if $A_{1}$ and $A_{2}$ join together. Similarly, for $G_{i}$ $(2\leq i\leq n-1)$, it follows from Lemma \ref{th:d1d2} that both $A_{2}$ and $A_{i+1}$ can only perform trivial OPLMs, and hence perfect discrimination of $G_{i}$ is possible only if $A_{2}$ and $A_{i+1}$ join together. The above $\sum_{i=2}^{n}(2d_{i}-1)$ states are indistinguishable in any bipartition and have genuine nonlocality. When $A_{1}$ performs the measurements $\{M^{A_{1}}_{1}=|d_{1}-1\rangle_{A_{1}} \langle d_{1}-1|,M^{A_{1}}_{2}=I_{A_{1}}-M^{A_{1}}_{1}\}$, the whole set can be locally reduced to two disjoint subsets \{$G_{1}, \cdots, G_{n-2}$\} and \{$G_{n-1}$\}. Therefore, the set consisting of the above states has genuine nonlocality of type \uppercase\expandafter{\romannumeral 1}.
\end{proof}

\section{Orthogonal product sets with genuine nonlocality of type \expandafter{\romannumeral 2} }

In this section, we construct type-\uppercase\expandafter{\romannumeral 2} genuinely nonlocal sets without entanglement, which are locally irreducible.

\subsection{Genuinely nonlocal sets in $\mathbb{C}^{d_{1}}\otimes\mathbb{C}^{d_{2}}\otimes\mathbb{C}^{d_{3}}~(3\leq d_{1}\leq d_{2}\leq d_{3})$}

In this subsection, we construct a genuinely nonlocal set of type \uppercase\expandafter{\romannumeral 2} in $\mathbb{C}^{d_{1}}\otimes\mathbb{C}^{d_{2}}\otimes\mathbb{C}^{d_{3}}~(3\leq d_{1}\leq d_{2}\leq d_{3})$.

\begin{proposition}
In $\mathbb{C}^3\otimes \mathbb{C}^4\otimes \mathbb{C}^5$, the following $14$ orthogonal product states are genuinely nonlocal:
\begin{equation}
\begin{aligned}
 |\phi_{1}\rangle=&|1\rangle_{A}|0-1\rangle_{B}|1\rangle_{C},\\
 |\phi_{2}\rangle=&|2\rangle_{A}|0-2\rangle_{B}|1\rangle_{C},\\
 |\phi_{3}\rangle=&|0-1\rangle_{A}|2\rangle_{B}|1\rangle_{C},\\
 |\phi_{4}\rangle=&|0-2\rangle_{A}|1\rangle_{B}|1\rangle_{C},\\
 |\phi_{5}\rangle=&|0-1\rangle_{A}|3\rangle_{B}|1\rangle_{C},\\
 |\phi_{6}\rangle=&|0+1\rangle_{A}|2-3\rangle_{B}|1\rangle_{C},\\
 |\phi_{7}\rangle=&|0+1+2\rangle_{A}|0+1+2+3\rangle_{B}|0+1+2+3+4\rangle_{C}, \\
 |\phi_{8}\rangle=&|1\rangle_{A}|0+1\rangle_{B}|0-1\rangle_{C},\\
 |\phi_{9}\rangle=&|2\rangle_{A}|0+1\rangle_{B}|0-2\rangle_{C},\\
 |\phi_{10}\rangle=&|0-1\rangle_{A}|0+1\rangle_{B}|2\rangle_{C},\\
 |\phi_{11}\rangle=&|0-1\rangle_{A}|0+1\rangle_{B}|3\rangle_{C},\\
 |\phi_{12}\rangle=&|0-1\rangle_{A}|0+1\rangle_{B}|4\rangle_{C},\\
 |\phi_{13}\rangle=&|0+1\rangle_{A}|0+1\rangle_{B}(|2\rangle_{C}+\omega_{3}|3\rangle_{C} +\omega_{3}^{2}|4\rangle_{C}),\\
 |\phi_{14}\rangle=&|0+1\rangle_{A}|0+1\rangle_{B}(|2\rangle_{C}+\omega_{3}^{2}|3\rangle_{C} +\omega_{3}|4\rangle_{C}).
 \end{aligned}
 \end{equation}
\end{proposition}

\begin{proof} By Lemma \ref{th:d1d2}, $\mathcal{A}_{2}=\{|1\rangle|0-1\rangle, |2\rangle|0-2\rangle, |0-1\rangle|2\rangle, |0-2\rangle|1\rangle, |0-1\rangle|3\rangle, |0+1\rangle|2-3\rangle, |0+1+2\rangle|0+1+2+3\rangle\}$ is a locally indistinguishable set in $\mathbb{C}^3\otimes\mathbb{C}^4$. $\mathcal{B}_{2}=\{|1\rangle|0-1\rangle, |2\rangle|0-2\rangle, |0-1\rangle|2\rangle, |0-2\rangle|1\rangle, |0-1\rangle|3\rangle, |0-1\rangle|4\rangle, |0+1\rangle(|2\rangle+\omega_{3}|3\rangle +\omega_{3}^{2}|4\rangle),|0+1\rangle (|2\rangle+\omega_{3}^{2}|3\rangle +\omega_{3}|4\rangle), |0+1+2\rangle|0+1+2+3+4\rangle\}$ is a locally indistinguishable set in $\mathbb{C}^3\otimes\mathbb{C}^5$.

As for the states in the subset $\{ |\phi_{i}\rangle\} _{i=1}^{7}$, the presence of $\mathcal{A}_{2}$ is between Alice and Bob. The states in the subsystem $C$ are not orthogonal. Then, even if Charlie comes together with either Alice or Bob, it is not possible to perfectly distinguish these states \cite{r13}, which leads to the local indistinguishability in partitions $B|CA$ and $A|BC$. For the subset  $\{ |\phi_{4}\rangle \bigcup \{ |\phi_{i}\rangle\} _{i=7}^{14} \}$, the presence of $\mathcal{B}_{2}$ between Alice and Charlie is evident. Bob's subsystems are not orthogonal. Then, even if Bob comes together with either Alice or Charlie, it is not possible to perfectly distinguish the states \cite{r13}, leading to the local indistinguishability in partitions $C|AB$ and $A|BC$. Thus the above set turns out to be genuinely nonlocal.
\end{proof}

Based on the above proposition, we put forward the following general construction in tripartite systems.

\begin{theorem}\label{The:3}
In $\mathbb{C}^{d_{1}}\otimes\mathbb{C}^{d_{2}}\otimes\mathbb{C}^{d{_3}}$ $(3\leq d_{1}\leq d_{2}\leq d_{3})$, the following $2d_{2}+2d_{3}-4$ orthogonal product states have genuine nonlocality of type \uppercase\expandafter{\romannumeral 2}:
\begin{equation} \footnotesize
\begin{aligned}
|\phi_{i}\rangle=&|i\rangle_{A}|0-i\rangle_{B}|1\rangle_{C},~1\leq i \leq d_{1}-1,\\
|\phi_{i+d_{1}-1}\rangle=&|0-i\rangle_{A}|j\rangle_{B}|1\rangle_{C},~1\leq i\leq d_{1}-2, \\&j=i+1;~i=d_{1}-1,~j=1,\\
|\phi_{m+d_{1}-1}\rangle=&|0-1\rangle_{A}|m\rangle_{B}|1\rangle_{C},~d_{1}\leq m\leq d_{2}-1,\\
|\phi_{s_{1}+d_{1}+d_{2}-2}\rangle=&|0+1\rangle_{A}(|2\rangle_{B}\\
&+\sum\limits_{t_{1}=1}^{d_{2}-d_{1}}\omega_{d_{2}-d_{1}+1}^{s_{1}t_{1}} |t_{1}+d_{1}-1\rangle_{B})|1\rangle_{C},\\
&1\leq s_{1}\leq d_{2}-d_{1},\\
|\phi_{2d_{2}-1}\rangle=&|0+1+\cdots+(d_{1}-1)\rangle_{A} \\&|0+1+\cdots+(d_{2}-1)\rangle_{B}\\&|0+1+\cdots+(d_{3}-1)\rangle_{C},\\
|\phi_{i+2d_{2}-1}\rangle=&|i\rangle_{A}|0+1\rangle_{B}|0-i\rangle_{C},~1\leq i\leq d_{1}-1,\\
|\phi_{i+d_{1}+2d_{2}-2}\rangle=&|0-i\rangle_{A}|0+1\rangle_{B}|j\rangle_{C},\\
&1\leq i\leq d_{1}-2,~ j=i+1,\\
|\phi_{n+d_{1}+2d_{2}-3}\rangle=&|0-1\rangle_{A} |0+1\rangle_{B}|n\rangle_{C},~d_{1}\leq n\leq d_{3}-1,\\
|\phi_{s_{2}+d_{1}+2d_{2}+d_{3}-4}\rangle=&|0+1\rangle_{A} |0+1\rangle_{B} \\&(|2\rangle_{C}+\sum\limits_{t_{2}=1}^{d_{3}-d_{1}}\omega_{d_{3}-d_{1}+1}^{s_{2}t_{2}} |t_{2}+d_{1}-1\rangle_{C}), \\&1\leq s_{2}\leq d_{3}-d_{1}.\\
\end{aligned}
\end{equation}
\end{theorem}

\begin{proof} Consider the states $\{ |\phi_{i}\rangle\} _{i=1}^{2d_{2}-1}$. It follows from Lemma \ref{th:d1d2} that there is a locally indistinguishable set in $\mathbb{C}^{d_{1}}\otimes\mathbb{C}^{d_{2}}$ between Alice and Bob. Furthermore, Charlie's states are all not orthogonal. As a result, even if Charlie comes together with either Alice or Bob, it is not possible to perfectly distinguish these states \cite{r13}, leading to the local indistinguishability in partitions $B|CA$ and $A|BC$.
For the states $\{|\phi_{2d_{1}-2}\rangle \bigcup \{ |\phi_{i}\rangle\} _{i=2d_{2}-1}^{2d_{2}+2d_{3}-4} \}$, there is a locally indistinguishable set in $\mathbb{C}^{d_{1}}\otimes\mathbb{C}^{d_{3}}$ between Alice and Charlie, and Bob's states are all not orthogonal. Then, even if Bob comes together with either Alice or Charlie, it is not possible to perfectly distinguish these states \cite{r13}, leading to the local indistinguishability in partitions $C|AB$ and $A|BC$. Consequently the above set turns out to be genuinely nonlocal.

Moreover, the above set is locally irreducible when all the parties are separated, since they can only perform trivial OPLMs. Hence, it is genuinely nonlocal of type \uppercase\expandafter{\romannumeral 2}.
\end{proof}

\subsection{Genuinely nonlocal sets in $\otimes^{n}_{i=1}\mathbb{C}^{d_{i}}$ $(3\leq d_{1}\leq d_{2}\leq \cdots\leq d_{n},n\geq4)$}

Next we consider the construction of orthogonal product states with genuine nonlocality of type \uppercase\expandafter{\romannumeral 2} in $\otimes^{n}_{i=1}\mathbb{C}^{d_{i}}$ $(3\leq d_{1}\leq d_{2}\leq \cdots\leq d_{n},n\geq4)$.

Denote the $2d_{2}-1$ locally indistinguishable orthogonal product states as $\{|x_{j_{1}}\rangle|y_{j_{1}}\rangle,j_{1}=1,2,\cdots,2d_{2}-1\}$ in $\mathbb{C}^{d_{1}}\otimes \mathbb{C}^{d_{2}}$ and  $2d_{k}-1$ locally indistinguishable orthogonal product states as $\{|x_{j_{k-1}}\rangle|z_{j_{k-1}}\rangle,j_{k-1}=1,2,\cdots,2d_{k}-1\}$ in $\mathbb{C}^{d_{2}}\otimes \mathbb{C}^{d_{k}}$ $(k=3,\cdots, n)$, respectively.

\begin{theorem}\label{The:4}
In $\otimes^{n}_{i=1}\mathbb{C}^{d_{i}}$ $(3\leq d_{1}\leq d_{2}\leq \cdots\leq d_{n}, n\geq4)$, the following set $\mathcal{S}_{3}$ of $\sum_{i=2}^{n}(2d_{i}-1)$ orthogonal product states have genuine nonlocality of type \uppercase\expandafter{\romannumeral 2}:
\begin{equation}
\begin{aligned}
G_{1}=&\{|x_{j_{1}}\rangle_{A_{1}}|y_{j_{1}}\rangle_{A_{2}}|0\rangle_{A_{3}}
|0\rangle_{A_{4}}\cdots\\&|0\rangle_{A_{n-3}}|0\rangle_{A_{n-2}}
|0\rangle_{A_{n-1}}|1\rangle_{A_{n}}\},\\
G_{2}=&\{|1\rangle_{A_{1}}|x_{j_{2}}\rangle_{A_{2}}|z_{j_{2}}\rangle_{A_{3}}
|0\rangle_{A_{4}}\cdots\\&|0\rangle_{A_{n-3}}|0\rangle_{A_{n-2}}|0\rangle_{A_{n-1}}
|0\rangle_{A_{n}}\},\\
G_{3}=&\{|0\rangle_{A_{1}}|x_{j_{3}}\rangle_{A_{2}}
|1\rangle_{A_{3}}|z_{j_{3}}\rangle_{A_{4}}\cdots\\&|0\rangle_{A_{n-3}}
|0\rangle_{A_{n-2}}|0\rangle_{A_{n-1}}|0\rangle_{A_{n}}\},\\
& ~~~~~~~~~~~~~~~~~~\vdots~~~~~~~~~~~~\\
G_{n-2}=&\{|0\rangle_{A_{1}}|x_{j_{n-2}}\rangle_{A_{2}}|0\rangle_{A_{3}}
|0\rangle_{A_{4}}\cdots\\&|0\rangle_{A_{n-3}}|1\rangle_{A_{n-2}}
|z_{j_{n-2}}\rangle_{A_{n-1}}|0\rangle_{A_{n}}\},\\
G_{n-1}=&\{|0\rangle_{A_{1}}|x_{j_{n-1}}\rangle_{A_{2}}|0\rangle_{A_{3}}
|0\rangle_{A_{4}}\cdots\\&|0\rangle_{A_{n-3}}|0\rangle_{A_{n-2}}
|1\rangle_{A_{n-1}}|z_{j_{n-1}}\rangle_{A_{n}}\}.
\end{aligned}
\end{equation}
\end{theorem}

\begin{proof} In fact, for $G_{1}$ both $A_{1}$ and $A_{2}$ can only perform trivial OPLMs by Lemma \ref{th:d1d2}. Thus perfect discrimination of $G_{1}$ can be achieved only if $A_{1}$ and $A_{2}$ join together. Similarly, for $G_{i}$ $(2\leq i\leq n-1)$ it follows from Lemma \ref{th:d1d2} that both $A_{2}$ and $A_{i+1}$ can only perform trivial OPLMs, and hence perfect discrimination of $G_{i}$ is possible only if $A_{2}$ and $A_{i+1}$ join together. Hence, the above $\sum_{i=2}^{n}(2d_{i}-1)$ states are indistinguishable in any bipartition, and the set $\mathcal{S}_{3}$ has genuine nonlocality.

Moreover, the above set of orthogonal product states is locally irreducible when all the parties are separated, because they can only perform trivial OPLMs. So $\mathcal{S}_{3}$ is genuinely nonlocal of type \uppercase\expandafter{\romannumeral 2}.
\end{proof}

\section{Conclusion}
The quantum nonlocality from the perspective of state discrimination is of great significance in quantum cryptographic protocols such as quantum secret sharing and quantum data hiding \cite{app1,app2,app3,app4,app5,app6,app7,app8}. In this paper, we have categorically presented the constructions of genuinely nonlocal sets without entanglement. For the type-\uppercase\expandafter{\romannumeral 1} genuinely nonlocal sets, we have put forward a construction in $\mathbb{C}^{d_{1}}\otimes\mathbb{C}^{d_{2}}\otimes\mathbb{C}^{d_{3}}~[3\leq (d_{1}-1)\leq d_{2}\leq d_{3}]$. Then we generalized it to \textit{n}-partite systems $\otimes^{n}_{i=1}\mathbb{C}^{d_{i}}$ $[3\leq (d_{1}-1)\leq d_{2}\leq\cdots\leq d_{n}, n\geq3]$ with size of $\sum_{i=2}^{n}(2d_{i}-1)$. For the type-\uppercase\expandafter{\romannumeral 2} genuinely nonlocal sets, we have presented a construction containing $2d_{2}+2d_{3}-4$ states in $\mathbb{C}^{d_{1}}\otimes\mathbb{C}^{d_{2}}\otimes\mathbb{C}^{d_{3}}$ $(3\leq d_{1}\leq d_{2}\leq d_{3})$. Furthermore, we have obtained a genuinely nonlocal set of type \uppercase\expandafter{\romannumeral 2} in $\otimes^{n}_{i=1}\mathbb{C}^{d_{i}}$ $(3\leq d_{1}\leq d_{2}\leq\cdots\leq d_{n},n\geq4)$ with size of $\sum_{i=2}^{n}(2d_{i}-1)$; see Table~ \ref{members} for these constructions of genuinely nonlocal sets without entanglement.
\begin{table}
	\newcommand{\tabincell}[2]{\begin{tabular}{@{}#1@{}}#2\end{tabular}}
	\centering
	\caption{\label{members}Genuinely nonlocal sets without entanglement.}
	\begin{tabular}{cccc}
		\toprule
		\hline
		\hline
		\specialrule{0em}{1.5pt}{1.5pt}
		~Reference~~&~~~System~~~&~~~Cardinality~~~&~~~Typology~~~\\
		\specialrule{0em}{1.5pt}{1.5pt}
		\midrule
		\hline
		\specialrule{0em}{1.5pt}{1.5pt}
		\tabincell{c}{\cite{r13}}&\tabincell{c}{$\mathbb{C}^{4}\otimes \mathbb{C}^{4}\otimes \mathbb{C}^{4}$}&\tabincell{c}{$64$}&\tabincell{c}{Type \uppercase\expandafter{\romannumeral 1}} \\
		\specialrule{0em}{3.5pt}{3.5pt}
		\tabincell{c}{\cite{r13}}&\tabincell{c}{$\mathbb{C}^{3}\otimes \mathbb{C}^{3}\otimes \mathbb{C}^{3}$}&\tabincell{c}{$27$}&\tabincell{c}{Type \uppercase\expandafter{\romannumeral 2}} \\
		\specialrule{0em}{3.5pt}{3.5pt}
		\tabincell{c}{\cite{r13}}&\tabincell{c}{$\mathbb{C}^{4}\otimes \mathbb{C}^{4}\otimes \mathbb{C}^{4}$}&\tabincell{c}{$64$}&\tabincell{c}{Type \uppercase\expandafter{\romannumeral 2}} \\
		\specialrule{0em}{3.5pt}{3.5pt}
		\tabincell{c}{\cite{r9}}&\makecell{\tabincell{c}{$\mathbb{C}^{x}\otimes \mathbb{C}^{y}\otimes \mathbb{C}^{z}$\\$(x,z\geq 3,y\geq 4)$}}&\makecell{\tabincell{c}{$2x+4y$\\$+2z-8$}}&\tabincell{c}{Type \uppercase\expandafter{\romannumeral 1}}\\
		\specialrule{0em}{3.5pt}{3.5pt}
		\tabincell{c}{\cite{r14}}&\tabincell{c}{$\mathbb{C}^{4}\otimes \mathbb{C}^{3}\otimes \mathbb{C}^{3}$}&\tabincell{c}{$14$}&\tabincell{c}{Type \uppercase\expandafter{\romannumeral 2}} \\
		\specialrule{0em}{3.5pt}{3.5pt}
		 \tabincell{c}{\cite{r14}}&\tabincell{c}{$\mathbb{C}^{m+2}\otimes(\mathbb{C}^{3})^{\otimes m}$}&\tabincell{c}{$6m+2$}&\tabincell{c}{Type \uppercase\expandafter{\romannumeral 2}} \\
		\specialrule{0em}{3.5pt}{3.5pt}
		\tabincell{c}{\cite{r14}}&\tabincell{c}{$\mathbb{C}^{6}\otimes \mathbb{C}^{5}\otimes \mathbb{C}^{5}$}&\tabincell{c}{$42$}&\tabincell{c}{Type \uppercase\expandafter{\romannumeral 2}} \\
        \specialrule{0em}{3.5pt}{3.5pt}
        \tabincell{c}{Theorem \ref{The:1}}&\tabincell{c}{$\mathbb{C}^{d_{1}}\otimes\mathbb{C}^{d_{2}}\otimes\mathbb{C}^{d_{3}}$\\$[3\leq (d_{1}-1)\leq d_{2}\leq d_{3}]$}&\tabincell{c}{$2(d_{2}+d_{3})-2$}&\tabincell{c}{Type \uppercase\expandafter{\romannumeral 1}} \\
		\specialrule{0em}{3.5pt}{3.5pt}
        \tabincell{c}{Theorem \ref{The:2}}&\tabincell{c}{$\otimes^{n}_{i=1}\mathbb{C}^{d_{i}}$}&\tabincell{c}{$\sum_{i=2}^{n}(2d_{i}-1)$}&\tabincell{c}{Type \uppercase\expandafter{\romannumeral 1}} \\
		\specialrule{0em}{3.5pt}{3.5pt}
        \tabincell{c}{Theorem \ref{The:3}}&\tabincell{c}{$\mathbb{C}^{d_{1}}\otimes\mathbb{C}^{d_{2}}\otimes\mathbb{C}^{d_{3}}$\\$(3\leq d_{1}\leq d_{2}\leq d_{3})$}&\tabincell{c}{$2(d_{2}+d_{3})-4$}&\tabincell{c}{Type \uppercase\expandafter{\romannumeral 2}} \\
		\specialrule{0em}{3.5pt}{3.5pt}
        \tabincell{c}{Theorem \ref{The:4}}&\tabincell{c}{$\otimes^{n}_{i=1}\mathbb{C}^{d_{i}}$}&\tabincell{c}{$\sum_{i=2}^{n}(2d_{i}-1)$}&\tabincell{c}{Type \uppercase\expandafter{\romannumeral 2}} \\
		\bottomrule
		\specialrule{0em}{1.5pt}{1.5pt}
		\hline
		\hline
	\end{tabular}
\end{table}

As the subset of a locally distinguishable set is locally distinguishable, while the superset of a locally indistinguishable set is locally indistinguishable, it is more interesting to explore the construction with small cardinality. Compared with the type-\uppercase\expandafter{\romannumeral 1} genuinely nonlocal set of Li \textit{et al}. in Ref. \cite{r9}, our construction contains fewer orthogonal product states in general tripartite systems. In particular, the genuinely nonlocal set in $\mathbb{C}^{4}\otimes\mathbb{C}^{3}\otimes\mathbb{C}^{3}$ has only ten states, which is less than the ones given by Rout \textit{et al.} in Ref. \cite{r14}. However, there are still many interesting problems that deserve further consideration, for instance, how to characterize the bound of a genuine nonlocal set and how to quantify the size of the nonlocality for a nonlocal set.

\section*{Acknowledgments}

This work is supported by the Basic Research Project of Shijiazhuang Municipal Universities in Hebei Province (Grant No. 241790697A), the Natural Science Foundation of Hebei Province (Grant No. F2021205001), the NSFC (Grants No. 62272208, No. 12171044, and No. 12075159), and the specific research fund of the Innovation Platform for Academicians of Hainan Province.

\appendix

\section{Proof of Proposition \ref{eg:35}} \label{app:1}

\begin{proof}
In order to distinguish these states, at least one part performs nontrivial measurements; that is, not all measurement operators $E_{m}$ are proportional to the identity operator, and it is necessary to preserve the orthogonality of the postmeasurement states for further discrimination.

$(i)$ When Alice performs orthogonality-preserving local measurements, the corresponding POVM elements can be expressed as $E_{1}=(a_{i,j})_{i,j \in Z_{3}}$. For $|\psi\rangle,|\phi\rangle\in \mathcal{S}$, $|\psi\rangle\neq |\phi\rangle$, we have
\begin{equation}\label{eq:9}
   \langle \psi| E_1\otimes I_2|\phi\rangle=0.
\end{equation}

First, with respect to $|\phi_{1}\rangle$ and $|\phi_{2}\rangle$, we have $\langle1|E_{1}|2\rangle \langle0-1|I_{2}|0-2\rangle =0$, i.e., $\langle1|E_{1}|2\rangle=0$. Therefore, $a_{1,2}=a_{2,1}=0$.

For $|\phi_{1}\rangle$ and $|\phi_{4}\rangle$, we have $\langle1|E_{1}|0-2\rangle \langle0-1|I_{2}|1\rangle =0$, i.e., $\langle1|E_{1}|0\rangle=\langle1|E_{1}|2\rangle=0$. Hence, $a_{1,0}=a_{0,1}=0$.

For $|\phi_{2}\rangle$ and $|\phi_{3}\rangle$, we have $\langle2|E_{1}|0-1\rangle \langle0-2|I_{2}|2\rangle =0$, i.e., $\langle2|E_{1}|0\rangle=\langle2|E_{1}|1\rangle=0$. Then $a_{2,0}=a_{0,2}=0$.

For $|\phi_{3}\rangle$ and $|\phi_{9}\rangle$, we have $\langle0-1|E_{1}|0+1+2\rangle \langle 2|I_{2}|0+1+2+3+4\rangle =0$, i.e., $\langle0-1|E_{1}|0+1+2\rangle=0$, $\langle0|E_{1}|0\rangle-\langle1|E_{1}|1\rangle=0$. So $\langle0|E_{1}|0\rangle=\langle1|E_{1}|1\rangle$. Similarly, for $|\phi_{4}\rangle$ and $|\phi_{9}\rangle$, it is easy to verify that $\langle0|E_{1}|0\rangle=\langle2|E_{1}|2\rangle$. Thus, $a_{0,0}=a_{1,1}=a_{2,2}$.

Therefore, we have $E_{1}\propto I$, which implies that Alice cannot start with a nontrivial measurement.

$(ii)$ When Bob performs orthogonality-preserving local measurements, the corresponding POVM elements can be expressed as $E_{2}=(b_{i,j})_{i,j \in Z_{5}}$. For $|\psi\rangle,|\phi\rangle\in \mathcal{S}$, $|\psi\rangle\neq |\phi\rangle$, we have
	\begin{equation}\label{eq:10}
		\langle \psi| I_1\otimes E_2|\phi\rangle=0.
	\end{equation}

Let us first consider $|\phi_{3}\rangle$ and $|\phi_{4}\rangle$. We have $\langle0-1|I_{1}|0-2\rangle \langle2|E_{2}|1\rangle =0$, i.e., $\langle2|E_{2}|1\rangle=0$, then $b_{1,2}=b_{2,1}=0$.

Next, for $|\phi_{1}\rangle$ and $|\phi_{3}\rangle$, we have $\langle1|I_{1}|0-1\rangle \langle0-1|E_{2}|2\rangle =0$, i.e., $\langle0|E_{2}|2\rangle=\langle1|E_{2}|2\rangle=0$, then $b_{0,2}=b_{2,0}=0$.

For $|\phi_{2}\rangle$ and $|\phi_{4}\rangle$, we have $\langle2|I_{1}|0-2\rangle \langle0-2|E_{2}|1\rangle =0$, i.e., $\langle0|E_{2}|1\rangle=\langle2|E_{2}|1\rangle=0$, so $b_{0,1}=b_{1,0}=0$.

For $|\phi_{5}\rangle$ and $|\phi_{6}\rangle$, we have $\langle0-1|I_{1}|0-1\rangle \langle3|E_{2}|4\rangle =0$, i.e., $\langle3|E_{2}|4\rangle=0$, then $b_{3,4}=b_{4,3}=0$. Similarly for $|\phi_{3}\rangle$ and $|\phi_{5}\rangle$, $|\phi_{3}\rangle$ and $|\phi_{6}\rangle$, we derive $b_{2,3}=b_{3,2}=0$ and $b_{2,4}=b_{4,2}=0$, respectively.

For $|\phi_{4}\rangle$ and $|\phi_{5}\rangle$, we have $\langle0-2|I_{1}|0-1\rangle \langle1|E_{2}|3\rangle =0$, i.e., $\langle1|E_{2}|3\rangle =0$, then $b_{1,3}=b_{3,1}=0$. Similarly for $|\phi_{4}\rangle$ and $|\phi_{6}\rangle$, a simple calculation gives $b_{1,4}=b_{4,1}=0$.

From $|\phi_{1}\rangle$ and $|\phi_{5}\rangle$, $|\phi_{1}\rangle$ and $|\phi_{6}\rangle$, it is easy to obtain that $\langle0-1|E_{2}|3\rangle=0$ as well as $\langle0-1|E_{2}|4\rangle=0$, which leads to $b_{0,3}=b_{3,0}=b_{1,3}=0$ and $b_{0,4}=b_{4,0}=b_{1,4}=0$, respectively.

For $|\phi_{1}\rangle$ and $|\phi_{9}\rangle$, we have $\langle1|I_{1}|0+1+2\rangle \langle0-1|E_{2}|0+1+2+3+4\rangle =0$, i.e., $\langle0-1|E_{2}|0+1+2+3+4\rangle=0$. So $\langle0|E_{2}|0\rangle=\langle1|E_{2}|1\rangle$, then $b_{0,0}=b_{1,1}$. Similarly, for $|\phi_{2}\rangle$ and $|\phi_{9}\rangle$, we obtain $b_{0,0}=b_{2,2}$. By $|\phi_{9}\rangle$ and $|\phi_{7}\rangle$, $|\phi_{9}\rangle$ and $|\phi_{8}\rangle$, we have $\langle0+1+2+3+4|E_{2}(|2\rangle+\omega_{3}|3\rangle+\omega_{3}^{2}|4\rangle)=0$ and  $\langle0+1+2+3+4|E_{2}(|2\rangle+\omega_{3}^{2}|3\rangle+\omega_{3}|4\rangle)=0$. Then,
\begin{equation}
\begin{aligned}
b_{22}+\omega_{3}b_{33}+\omega_{3}^{2}b_{44}=&0,\\
b_{22}+\omega_{3}^{2}b_{33}+\omega_{3}b_{44}=&0.
\end{aligned}
\end{equation}

By using the following fact from linear algebra that if
\begin{equation}
\begin{aligned}
x_{1}+x_{2}+\cdots+x_{n}=&0,\\
x_{1}+\omega_{n}x_{2}+\cdots+\omega_{n}^{(n-1)}x_{n}=&0,\\
x_{1}+\omega_{n}^{2}x_{2}+\cdots+\omega_{n}^{2(n-1)}x_{n}=&0,\\
~~~~~~~~~~~~~~\vdots~~~~~~~~~~~~~~~~~~\\
x_{1}+\omega_{n}^{n-1}x_{2}+\cdots+\omega_{n}^{(n-1)(n-1)}x_{n}=&0,
\end{aligned}
\end{equation}
then $x_{1}=x_{2}=\cdots=x_{n}=0$, and if
\begin{equation}
\begin{aligned}
x_{1}+\omega_{n}x_{2}+\cdots+\omega_{n}^{(n-1)}x_{n}=&0,\\
x_{1}+\omega_{n}^{2}x_{2}+\cdots+\omega_{n}^{2(n-1)}x_{n}=&0,\\
~~~~~~~~~~~~\vdots~~~~~~~~~~~~~~~~~~~\\
x_{1}+\omega_{n}^{n-1}x_{2}+\cdots+\omega_{n}^{(n-1)(n-1)}x_{n}=&0,
\end{aligned}
\end{equation}
then $x_{1}=x_{2}=\cdots=x_{n}$, where $w_n=e^{\frac{2\pi\sqrt{-1}}{n}}$, and we obtain $b_{2,2}=b_{3,3}=b_{4,4}$. Hence, $b_{0,0}=b_{1,1}=b_{2,2}=b_{3,3}=b_{4,4}$.
Consequently, we have $E_{2}\propto I$, which means that Bob cannot start with a nontrivial measurement either.

In summary, both Alice and Bob can only start with a trivial OPLM, which implies that the nine orthogonal product states in Proposition $\ref{eg:35}$ are locally indistinguishable.
\end{proof}

\section{Proof of Lemma \ref{th:d1d2}} \label{app:2}

\begin{proof}
We need to show that neither Alice nor Bob can perform nontrivial measurements.

$(i)$ When Alice performs orthogonality-preserving local measurements, the corresponding POVM elements can be expressed as $E_{1}=(a_{i,j})_{i,j\in Z_{d_{1}}}$. For $|\psi\rangle,|\phi\rangle\in \mathcal{S}$, $|\psi\rangle\neq |\phi\rangle$, we have
	\begin{equation}\label{eq:12}
		\langle \psi| E_1\otimes I_2|\phi\rangle=0.
	\end{equation}

First, with respect to $|\phi_{i}\rangle$ and $|\phi_{i'}\rangle$, where $1\leq i\neq i' \leq d_{1}-1$, we have $\langle i|E_{1}|i'\rangle \langle0-i|I_{2}|0-i'\rangle=0$, i.e., $\langle i|E_{1}|i'\rangle=0$. Therefore, $a_{i,i'}=a_{i',i}=0$ for $1\leq i\neq i' \leq d_{1}-1$.

For $|\phi_{m}\rangle$ and $|\phi_{i+(d_{1}-1)}\rangle$, where $ 1\leq i\leq d_{1}-2$, $m=i+1$; $i=d_{1}-1$, $m=1$, we have $\langle m|E_{1}|0-i\rangle \langle0-m|I_{2}|m\rangle =0$, i.e., $\langle m|E_{1}|0-i\rangle=0$. Then $\langle m|E_{1}|0\rangle=\langle m|E_{1}|i\rangle=0$, from which we have $a_{0,m}=a_{m,0}=0$, $1\leq m \leq d_{1}-1$.

Last, for $|\phi_{i+(d_{1}-1)}\rangle$ and $|\phi_{2d_{2}-1}\rangle $, where $ 1\leq i \leq d_{1}-2,~m=i+1$; $i=d_{1}-1,~m=1$, we have $\langle 0-i|E_{1}|0+\cdots+(d_{1}-1)\rangle \langle m|I_{2}|0+\cdots+(d_{2}-1)\rangle =0$. Then $\langle 0-i|E_{1}|0+\cdots+(d_{1}-1)\rangle=0$, i.e., $\langle 0|E_{1}|0\rangle=\langle i|E_{1}|i\rangle$, and thus $a_{0,0}=a_{i,i}$ for $1\leq i \leq d_{1}-1$.

Thus $E_{1}\propto I$, which implies that Alice cannot perform nontrivial measurements.

$(ii)$ When Bob performs orthogonality-preserving local measurements, the corresponding POVM elements can be expressed as $E_{2}=(b_{i,j})_{i,j \in Z_{d_{2}}}$. For $|\psi\rangle,|\phi\rangle\in \mathcal{S}$, $|\psi\rangle\neq |\phi\rangle$, we have
	\begin{equation}\label{eq:13}
		\langle \psi| I_1\otimes E_2|\phi\rangle=0.
	\end{equation}

First, for $|\phi_{i+(d_{1}-1)}\rangle$ and $|\phi_{i'+(d_{1}-1)}\rangle$, where $i\neq i',~m\neq m'$, we get $\langle0-i|I_{1}|0-i'\rangle \langle m|E_{2}|m'\rangle =0$, i.e., $\langle m|E_{2}|m'\rangle=0$. Then $b_{m,m'}=b_{m',m}=0$ for $1\leq m\neq m'\leq d_{1}-1$.

For $|\phi_{j+(d_{1}-1)}\rangle$ and $|\phi_{j'+(d_{1}-1)}\rangle$, $d_{1}\leq j\neq j' \leq d_{2}-1$, similarly we have $\langle0-1|I_{1}|0-1\rangle \langle j|E_{2}|j'\rangle =0$, i.e., $\langle j|E_{2}|j'\rangle=0$. Then $b_{j,j'}=b_{j',j}=0$ for $d_{1}\leq j\neq j' \leq d_{2}-1$.

Next from the states $|\phi_{i}\rangle$ and $|\phi_{i+(d_{1}-1)}\rangle$, where $ 1\leq i \leq d_{1}-1,~1\leq m \leq d_{1}-1$, we have $\langle i|I_{1}|0-i\rangle \langle 0-i|E_{2}|m\rangle =0$, which means $\langle 0|E_{2}|m\rangle=\langle i|E_{2}|m\rangle$. Hence, $b_{0,m}=b_{m,0}=b_{i,m}=0,~1\leq i \leq d_{1}-1$ for $1\leq m\leq d_{1}-1$.

For $|\phi_{i+(d_{1}-1)}\rangle$ and $|\phi_{j+(d_{1}-1)}\rangle$, $1\leq i \leq d_{1}-1,~d_{1}\leq j \leq d_{2}-1$, it is direct to find that $\langle0-i|I_{1}|0-1\rangle \langle m|E_{2}|j\rangle=0$. Then we have $\langle m|E_{2}|j\rangle=0$, which means $b_{m,j}=b_{j,m}=0$ for $1\leq m \leq d_{1}-1$ and $d_{1}\leq j \leq d_{2}-1$.

For $|\phi_{1}\rangle$ and $|\phi_{j+(d_{1}-1)}\rangle$, $d_{1}\leq j \leq d_{2}-1$, we have the relation $\langle 1|I_{1}|0-1\rangle \langle0-1|E_{2}|j\rangle =0$, i.e., $\langle 0|E_{2}|j\rangle=\langle 1|E_{2}|j\rangle$. So $b_{0,j}=b_{j,0}=b_{1,j}=0$ for $d_{1}\leq j \leq d_{2}-1$.

For $|\phi_{i}\rangle$ and $|\phi_{2d_{2}-1}\rangle$, $1\leq i \leq d_{1}-1$, we have $\langle i|I_{1}|0+\cdots+(d_{1}-1)\rangle \langle0-i|E_{2}|0+\cdots+(d_{2}-1)\rangle =0$, i.e., $\langle 0-i|E_{2}|0+\cdots+(d_{2}-1)\rangle=0$. Then $b_{0,0}=b_{i,i}$ for $1\leq i \leq d_{1}-1$.

For $|\phi_{2d_{2}-1}\rangle$ and $|\phi_{s+(d_{1}+d_{2}-2)}\rangle$, $1\leq s \leq d_{2}-d_{1}$, we have $\langle 0+\cdots+(d_{1}-1)|I_{1}|0+1\rangle \langle0+\cdots+(d_{2}-1)|E_{2}(|2\rangle+\sum\limits_{t=1}^{d_{2}-d_{1}}
\omega_{d_{2}-d_{1}+1}^{st} |t+d_{1}-1\rangle) =0$, i.e., $\langle 0+\cdots+(d_{2}-1)|E_{2}(|2\rangle+\sum\limits_{t=1}^{d_{2}-d_{1}}
\omega_{d_{2}-d_{1}+1}^{st}|t+d_{1}-1\rangle)=0$. Then
\begin{equation}\small
\begin{aligned}
b_{2,2}+\omega_{d_{2}-d_{1}+1}b_{d_{1},d_{1}}+\cdots+\omega_{d_{2}-d_{1}+1}^{d_{2}-d_{1}}b_{d_{2}-1,d_{2}-1}=&0,\\
b_{2,2}+\omega_{d_{2}-d_{1}+1}^{2}b_{d_{1},d_{1}}+\cdots+\omega_{d_{2}-d_{1}+1}^{2(d_{2}-d_{1})}b_{d_{2}-1,d_{2}-1}=&0,\\
~~~~~~~~~~~~~~~~~~~\vdots~~~~~~~~~~~~~~~~~~~~~~~~~~~~~~~~~\\
b_{2,2}+\omega_{d_{2}-d_{1}+1}^{d_{2}-d_{1}}b_{d_{1},d_{1}}+\cdots+\omega_{d_{2}-d_{1}+1}^{(d_{2}-d_{1})(d_{2}-d_{1})}b_{d_{2}-1,d_{2}-1}=&0.\\
\end{aligned}
\end{equation}

According to the results from linear algebra, we have $b_{2,2}=b_{d_{1},d_{1}}=\cdots=b_{d_{2}-1,d_{2}-1}$, namely, $E_{2}\propto I$, which means that Bob cannot perform nontrivial measurements either.

In summary, both Alice and Bob cannot perform nontrivial OPLM, so the $2d_{2}-1$ orthogonal product states in Lemma $\ref{th:d1d2}$ are locally indistinguishable.
\end{proof}
\end{document}